\documentclass[conference]{IEEEtran}
\usepackage{amsmath,amssymb,amsthm,bm,mathtools,mathrsfs}
\usepackage{graphicx}
\usepackage{enumitem}
\usepackage{mathrsfs}
\usepackage{enumitem}
\usepackage{stfloats}
\usepackage{graphicx}
\usepackage{cite}

\newtheorem{lemma}{Lemma}
\newtheorem{proposition}{Proposition}
\newtheorem{theorem}{Theorem}
\newtheorem{remark}{Remark}

\DeclareMathOperator{\tr}{tr}

\begin{document}
    
    \title{Optimal Movable Antenna Placement for Near-Field Wireless Sensing}
    
    \author{
    	\IEEEauthorblockN{Jinjian Liu,
    		Xianxin Song, and
    		Xianghao Yu}
    	\IEEEauthorblockA{Department of Electrical Engineering, City University of Hong Kong, Hong Kong}
    	\IEEEauthorblockA{Emails: jinjian.liu@my.cityu.edu.hk, xianxin.song@cityu.edu.hk, alex.yu@cityu.edu.hk}
    }

    \maketitle
    
	\begin{abstract}
		Movable antennas (MAs) have emerged as a promising technology for wireless sensing by reconfiguring antenna positions to exploit additional spatial degrees of freedom (DoFs). This paper investigates a robust movable antenna placement strategy for near-field wireless sensing to minimize the worst-case squared position error bound (SPEB). By temporarily relaxing the minimum inter-element spacing constraint, we first establish the optimality of centro-symmetric antenna position distribution, which simplifies the identification of the worst-case source, locating it at the array broadside on the Rayleigh boundary. Moreover, by leveraging moment-based analysis with the Richter-Tchakaloff theorem, we derive a closed-form optimal solution with three points supported on the center and two edges of the array. Guided by this structural insight, we finally develop an efficient three-point discrete deployment strategy to ensure the minimum inter-element spacing. Simulations demonstrate that the proposed design consistently outperforms conventional fixed antenna arrays and matches the exhaustive search benchmark at negligible computational complexity.
	\end{abstract}

\section{Introduction}
    
    Movable antenna (MA) has emerged as a key enabling technology for sixth-generation (6G) wireless networks to fulfill the stringent requirements of ultra-reliable and high-precision wireless sensing. Unlike the traditional antenna configurations with uniform distribution \cite{RobertsTAP2011}, MAs can reconfigure the electromagnetic channel conditions by adjusting the antenna positions within a spatial region\cite{ZhuCST}. This capability unlocks additional spatial degrees of freedom (DoFs), offering great potential to enhance wireless sensing performance \cite{ZhuCST}.

    To fully exploit this potential, MA systems typically employ large array apertures, thereby effectively enlarging the near-field region and rendering spherical-wave effects non-negligible \cite{liuOJCS2023}.
    Driven by this physical reality, MA-enabled sensing in the near-field region has attracted significant research attention \cite{GazzahTAP2014, wang202512arxiv, DingTWC2026, SunIoTJ2025}. 
    Specifically, authors in \cite{GazzahTAP2014} derived the Cram{\'e}r-Rao bound (CRB) expression with respect to antenna positions in the near-field scenario, numerically demonstrated that centro-symmetric array geometries can improve both range and angle estimation performance compared to traditional uniform-linear-array (ULA).
    Furthermore, prior work \cite{wang202512arxiv} adopted worst-case CRB as the estimation performance metric, and accordingly optimized the antenna location to minimize the worst-case estimation CRB in the near-field region by utilizing an iterative algorithm. 

    Moreover, the potential of MA-enabled sensing has been explored to enhance near-field integrated sensing and communications (ISAC) performance in \cite{DingTWC2026, SunIoTJ2025}. 
    Specifically, \cite{DingTWC2026} maximized a weighted sum of sensing mutual information and communication rate by jointly optimizing the MA positions, sensing-signal covariance
    matrices and transmit/receive beamformers.
    Similarly, \cite{SunIoTJ2025} investigated rotatable MAs to balance the angle-range CRB and communication sum-rate through the joint design of transmit beamformers and the rotatable MA positions/rotations.
    However, current studies on near-field MA sensing \cite{wang202512arxiv} and ISAC \cite{DingTWC2026, SunIoTJ2025} primarily rely on iterative algorithms to optimize the antenna position variables, which obscures structural insights of the optimal array geometry and incurs extremely high computational costs, especially for large-scale antenna arrays.

    In this paper, we investigate near-field device-based sensing enabled by linear MA arrays, aiming to characterize the closed-form optimal antenna geometry for minimizing the worst estimation error of the source's position. We first establish a near-field sensing model and adopt the squared position error bound (SPEB) as the sensing performance metric\cite{Wintit2010} to avoid the dimensional mismatch between angle and range error.
    By temporarily relaxing the minimum inter-element spacing constraint, we prove that the optimal antenna position distribution is centro-symmetric and further demonstrate that the array broadside on the Rayleigh boundary yields the worst-case SPEB. Furthermore, Richter-Tchakaloff theorem is leveraged to reveal that the optimal solution can be characterized by a three-point centro-symmetric structure, leading to a practical discrete deployment strategy that satisfies the spacing constraint. Simulation results demonstrate that the proposed design significantly outperforms traditional fixed-position antenna arrays and achieves performance comparable to exhaustive search with orders-of-magnitude lower computational complexity.      
    
    \textit{Notations:}
    We use normal-face letters to denote scalars,
    boldface lowercase and uppercase letters to denote column vectors and matrices, respectively.
    The transpose is denoted by $(\cdot)^{\mathrm T}$.
    The imaginary unit is denoted by $\jmath = \sqrt{-1}$.
    For a vector $\bm x$, $x_n$ denotes its $n$-th entry.
    The $N\times N$ identity matrix is denoted by $\mathbf I_N$.
    The set of $M \times N$ real matrices is denoted by $\mathbb{R}^{M \times N}$.
    The set of complex numbers is denoted by $\mathbb{C}$.
    The trace is denoted by $\mathrm{tr}(\cdot)$.
    The variance and covariance operations are denoted by $\mathrm{Var}(\cdot)$ and $\mathrm{Cov}(\cdot,\cdot)$, respectively.
    The notation $\arg\max$ denotes the set of maximizers.
    A circularly symmetric complex Gaussian vector with mean $\bm \mu$ and covariance $\mathbf A$
    is denoted by $\mathcal{CN}(\bm \mu,\mathbf A)$.
    The set of $n\times n$ real symmetric positive definite matrices is denoted by $\mathbb{S}^{n}_{++}$,
    and $\preceq$/$\succeq$ denotes the Loewner partial order.

    \section{System Model and Problem Formulation}
    \label{sec:system_model}
	We consider a single-source near-field sensing scenario, as shown in Fig.~\ref{fig:nearfield_scenario}. The system consists of a radiating source acting as the transmitter and a linear MA array serving as the sensing receiver. In this work, the sensing region of interest is restricted in the radiating near-field region of the MA \cite{liuOJCS2023}, i.e., the Fresnel region, which is denoted by $\Gamma$. 
	Over $\Gamma$, we adopt the uniform spherical wave (USW) channel model in \cite{liuOJCS2023}. To ensure the validity of the USW model, we restrict the source range to $r \in [d_{\mathrm{F}}, d_{\mathrm{R}}]$, where $d_{\mathrm{F}}$ and $d_{\mathrm{R}}$ denote the Fresnel and Rayleigh distances, respectively.
	
	\begin{figure}[t]
		\centering
		\includegraphics[width=0.65\linewidth, height=0.65\linewidth, keepaspectratio]{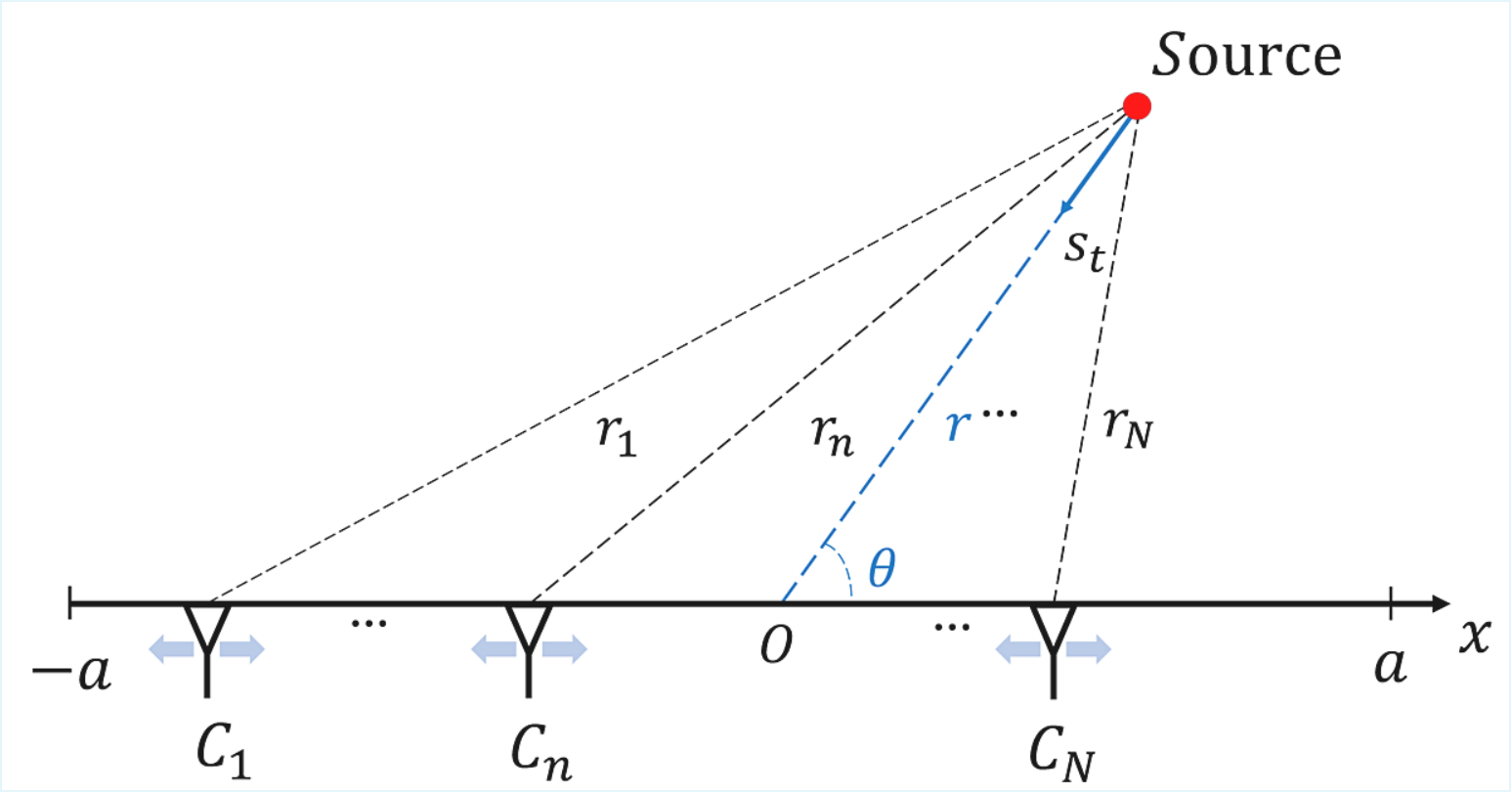}
		\caption{The 1D MA array geometry for single-source sensing in the near field.}
		\label{fig:nearfield_scenario}
	\end{figure}
	
	Specifically, the linear MA array is equipped with $N$ receive antennas aligned along the $x$-axis. Let $C_n = (x_n, 0)$ denote the location of $n$-th receive antenna with $x_n\in[-a,a]$, $n=1,\ldots,N$, where the aperture is given by $D=2a$. Let
	$\bm{x}=[x_1,\ldots,x_N]^{\mathrm{T}}$ collect the antenna coordinates. Without loss of generality, the antennas are indexed in ascending order of their locations, i.e., $-a \le x_1 < x_2 \dots < x_N \le a$.
	Let $\theta\in[0,\pi]$ denote the angle of arrival (AoA), defined as the angle between the source-to-array-center line and the positive $x$-axis.
	Accounting for the effective aperture loss \cite{liuOJCS2023,LuTWC2022}, we define the effective aperture as $D_{\mathrm{eff}} = D\sin\theta$ for any given $\theta$. Accordingly, the Fresnel and Rayleigh distances are implicitly represented by $d_{\mathrm{F}} = 0.62\sqrt{\frac{D_{\mathrm{eff}}^3}{\lambda}}$  and $d_{\mathrm{R}} = \frac{2D_{\mathrm{eff}}^2}{\lambda}$, respectively \cite{Balanis}.
	
	The source emits a narrowband signal $s_t$ with wavelength $\lambda$ toward the MA array, where $t \in \{1, \dots, T\}$ denotes the snapshot index and $T$ is the total number of snapshots. The signal $s_t$ is modeled as a deterministic sensing waveform satisfying $\frac{1}{T} \sum_{t=1}^{T} |s_t|^2 = P$, with $P$ representing the transmit power. The position of the source can be characterized either by its Cartesian coordinates
		\begin{equation}
			\bm{p}=[p_1,p_2]^{\mathrm{T}},
		\end{equation}
	or, equivalently, by the polar coordinates
$
		\bm{\eta}=[u,r]^{\mathrm{T}},
$
	where $u=\cos\theta$ and $r = \sqrt{p_1^2+p_2^2}$ represents the range from the array origin to the source.

    Denote $r_n$ as the propagation distance between the source and the
    $n$-th antenna:
    	\begin{equation}
    		r_n=\sqrt{r^2-2rux_n+x_n^2}.
    		\label{eq:rn}
    	\end{equation}	
    Using the Fresnel approximation \cite{liuOJCS2023} of \eqref{eq:rn} and after removing the common range-dependent phase term $e^{\jmath\frac{2\pi}{\lambda}r}$, the array manifold vector at the transmitter can be written as \cite{wang202512arxiv}
    \begin{equation}
    	\label{eq. steering vector}
    	\boldsymbol{\alpha}(\bm{x},\bm{\eta}) \approx
    	\left[ 
    	e^{\jmath \frac{2\pi}{\lambda} \left(ux_1-\frac{1-u^2}{2r}x_1^2\right)}, 
    	\dots, 
    	e^{\jmath \frac{2\pi}{\lambda} \left(ux_N-\frac{1-u^2}{2r}x_N^2\right)} 
    	\right]^{\mathrm{T}}.
    \end{equation}
    
	Under the adopted USW model, which assumes uniform path loss amplitudes across the array elements while capturing the phase variations, we now formulate the received signal. Specifically, the received signal at the MA array collected over $T$ snapshots can be expressed as
	\begin{equation}
		\bm{y}_t=\beta\,\boldsymbol{\alpha}(\bm{x},\bm{\eta})\,s_t+\bm{n}_t,\qquad
		t=1,\ldots,T,
	\end{equation}
	where $\beta\in\mathbb{C}$ is a unified channel coefficient, and
	$\bm{n}_t\sim\mathcal{CN}(\bm{0},\sigma^2\mathbf{I}_N)$ is the additive white Gaussian noise (AWGN) at the sensing receiver.
	We stack the transmit signals, receive signals, and the noise across $T$ snapshots as $\boldsymbol{s} = [s_1, \ldots, s_T]^{\mathrm{T}}$, $\bm{Y} = [\boldsymbol{y}_1, \ldots, \boldsymbol{y}_T]$, and $\bm{N} = [\boldsymbol{n}_1, \ldots, \boldsymbol{n}_T]$, respectively, and accordingly,
		\begin{equation}
			\begin{aligned}
				\label{eq. stacked receive matrix}
				\bm{Y} = \boldsymbol{\alpha}(\bm{x},\bm{\eta})\boldsymbol{s}^{\mathrm{T}}+\bm{N}.
			\end{aligned}
		\end{equation}
	
	  To evaluate the fundamental limits of the source position estimation performance, let $\mathrm{CRB}_{\bm{\eta}} \in \mathbb {R}^{2\times2}$ denote the CRB matrix for estimating $\bm{\eta}$. Based on the received signal model in \eqref{eq. stacked receive matrix}, the diagonal entries of $\mathrm{CRB}_{\bm{\eta}}$ representing the theoretical variance lower bounds for the parameters $u$ and $r$, are denoted by $\mathrm{CRB}_u$ and $\mathrm{CRB}_r$, respectively, given by \cite{wang202512arxiv}
	\begin{equation}
		\begin{aligned}
			\label{eq. polar CRB}
			\mathrm{CRB}_u(\bm{x})
			&=
			\kappa\,
			\frac{G_1(\tilde{\bm{x}})}
			{G_1(\bm{x})\, G_1(\tilde{\bm{x}})
				-{G_2^2}(\bm{x}, \tilde{\bm{x}})},\\[0.4ex]
			\mathrm{CRB}_r(\bm{x}, \bm{\eta})
			&=
			\kappa\,
			\frac{
				4r^4 G_1(\bm{x})
				+ 8 u r^3 G_2(\bm{x}, \tilde{\bm{x}})
				+ 4 u^2 r^2 G_1(\tilde{\bm{x}})
			}{
				(1 - u^2)^2
				\bigl(
				G_1(\bm{x}) G_1(\tilde{\bm{x}})
				-{G_2^2}(\bm{x}, \tilde{\bm{x}})
				\bigr)
			}.
		\end{aligned}
	\end{equation}
	Here
	$\kappa = \frac{\sigma^2 \lambda^2}{8\pi^2 T P N |\beta|^2}$ and
	$\tilde{\bm{x}} = [x_1^2, \dots, x_N^2]^{\mathrm{T}}$. The $G_1(\cdot)$ and $G_2(\cdot,\cdot)$ functions are defined as $G_1(\boldsymbol{x}) = \frac{1}{N} \sum_{n=1}^N x_n^2 - g(\boldsymbol{x})^2$ and $G_2(\boldsymbol{x}, \tilde{\boldsymbol{x}}) = \frac{1}{N} \sum_{n=1}^N x_n^3 - g(\boldsymbol{x})g(\tilde{\boldsymbol{x}})$, respectively, with $g(\boldsymbol{x}) = \frac{1}{N} \sum_{n=1}^N x_n$.

    \begin{remark}\label{remark:SPEB}
    	It is worth noting that directly minimizing the trace of the CRB matrix with respect to $\bm{\eta} = [u, r]^{\mathrm{T}}$ involves mismatched units (i.e., the dimensionless $u$ versus $r$ in meters), leading to a physically ambiguous objective. This motivates us to consider a dimensionally unified metric.
    \end{remark}
    
    To evaluate the sensing performance for the source's position, we adopt the SPEB as the metric \cite{Wintit2010}. The SPEB serves as the theoretical lower bound on the mean squared position error in Cartesian coordinates $\boldsymbol{p}$ for any unbiased estimator. Mathematically, it is defined as the trace of the Cartesian CRB matrix and denoted by $\mathrm{SPEB}(\bm{x},\bm{p})$, i.e., $\mathrm{SPEB}(\bm{x},\bm{p}) = \tr(\mathrm{CRB}(\bm{x},\bm{p}))$ \cite{Wintit2010}. Its relationship with the polar CRB is given in the following lemma.
    
    \begin{lemma}\label{lemma:CRB_Cartesian}
    	The source SPEB is given by
    	\begin{equation}
    		\begin{aligned}
    			\label{eq. CRB-transform}
    			\mathrm{SPEB}(\bm{x},\bm{p}) =
    			\frac{r^2}{1-u^2}\,\mathrm{CRB}_u(\bm{x})
    			+
    			\mathrm{CRB}_r(\bm{x},\bm{\eta}).
    		\end{aligned}
    	\end{equation}
    \end{lemma}
    \begin{IEEEproof}
    	See Appendix~A.
    \end{IEEEproof}

	In this paper, we aim to ensure robust performance against the worst-case source location within the prescribed near-field region $\Gamma$. This is achieved by minimizing the maximum $\mathrm{SPEB}(\bm{x},\bm{p})$, i.e.,
     \begin{subequations}
     	\begin{alignat}{3}
     		\text{(P1)}: & \quad & \min_{\boldsymbol{x}} \quad \max_{\boldsymbol{p} \in \Gamma} & \quad \mathrm{SPEB}(\boldsymbol{x},\boldsymbol{p}) \\
     		& & \mathrm{s.t.} & \quad x_n \in [-a, a], \quad \forall n=1, \dots, N, \\
     		& & & \quad x_{n+1} - x_n \ge \frac{\lambda}{2}, \quad \forall n=1,\dots,N-1, \label{eq. cons-inter-spacing}
     	\end{alignat}
     \end{subequations}
     where constraint \eqref{eq. cons-inter-spacing} ensures minimum inter-element spacing to avoid mutual coupling. 
	
	\section{Problem Reformulation and Analysis}
	   
    In problem (P1), the minimum spacing constraint \eqref{eq. cons-inter-spacing} couples adjacent antenna positions, which makes the placement design mathematically intractable. To facilitate our analysis and gain theoretical insight, we temporarily relax \eqref{eq. cons-inter-spacing}. 
    Moreover, from \eqref{eq. polar CRB}, we observe that both $G_1(\bm x)$ and $G_2(\bm x,\tilde{\bm x})$ are closely related to the term $g(\boldsymbol{x}) = \frac{1}{N} \sum_{n=1}^N x_n$. This term can be naturally interpreted as an average over $N$ samples $\{x_n\}_{n=1}^N$, which are drawn from a certain distribution of a random variable $X$.
    Specifically, let $X$ be a discrete random variable supported on $\mathcal S\subseteq[-a,a]$ that models the distribution of antenna locations, and view elements in $\bm{x}$ as $N$ realizations of $X$. As $N \to \infty$, functions $G_1(\bm x)$ and $G_2(\bm x,\tilde{\bm x})$ can be identified with the variance and covariance associated with $X$, i.e., $\mathrm{Var}(X)$ and $\mathrm{Cov}(X,X^2)$, respectively. Therefore, we can reformulate the discrete antenna placement problem as a distribution design problem for $X$. To characterize the distribution of $X$, we introduce its probability mass function $w(\cdot)$ on $\mathcal S$ by $w(\zeta)= \Pr(X=\zeta)$ for all $\zeta\in\mathcal S$.

    This transformation leads to the reformulated problem $(\text{P1}')$ which is explicitly expressed as
    \begin{subequations}
    	\begin{alignat}{3}
    		(\text{P1}'): & \quad & \min_{w(\cdot)} \quad \max_{\boldsymbol{p} \in \Gamma} & \quad \mathrm{SPEB}(w(\cdot),\boldsymbol{p}) \\
    		& & \mathrm{s.t.} & \quad \sum_{\zeta\in\mathcal S} w(\zeta)=1, \\ &&& \quad w(\zeta)\ge 0,\quad \forall \zeta\in\mathcal S.
    	\end{alignat}
    \end{subequations}
    Here, $\mathrm{SPEB}(w(\cdot),\boldsymbol{p})$ denotes the SPEB evaluated under the distribution $w(\cdot)$.
     
    Based on the theoretical optimal solution for $(\text{P1}')$, we will finally construct a feasible discrete antenna placement for (P1) that strictly satisfies the spacing constraint.
	To solve $(\text{P1}')$, we adopt a two-stage approach to decouple the min-max structure. We first address the inner maximization stage by defining the worst-case design metric $F(w(\cdot))$ as
	\begin{equation}
		\label{eq:def-worstcase}
		F(w(\cdot))
		\;=\;
		\max_{{\boldsymbol{p}}\in \Gamma} \quad
		\mathrm{SPEB}(w(\cdot),\bm{p}).
	\end{equation}

	\begin{proposition}
		For any distribution $w(\cdot)$ supported on $[-a,a]$, let $w_{\mathrm{sym}}(\cdot)$ denote its symmetrized counterpart constructed as a mixture of $w(\cdot)$ and its reflected distribution. Then, the following inequality holds:
		\begin{equation}
			F(w_{\mathrm{sym}}(\cdot)) \le F(w(\cdot)).
		\end{equation}
	\end{proposition}
	
	\begin{IEEEproof}
		See Appendix~\ref{app.proposition1}.
	\end{IEEEproof}
	
	In light of Proposition 1, which implies that an optimal solution to the distribution design problem can be found within the class of distributions that are symmetric about the origin. Hence, we restrict the search to such centro-symmetric distributions. Under symmetry, the covariance term vanishes, i.e., $\mathrm{Cov}(X,X^2)=0$. Consequently, for any centro-symmetric distribution $w(\cdot)$, $\mathrm{SPEB}$ is reduced to
	\begin{equation}
		\label{eq:positive-coefficients}
		\begin{aligned}
			&\mathrm{SPEB}(w(\cdot),\bm{p}) = \kappa\!\left(
			\frac{{k_1}}{\mathrm{Var}(X)}
			+
			\frac{{k_2}}{\mathrm{Var}(X^2)}
			\right), 
		\end{aligned}
	\end{equation}
	where $k_1 = \frac{r^2(1+3u^2)}{(1-u^2)^2}$ and $
	k_2 = \frac{4r^4}{(1-u^2)^2}$. Then, $F(w(\cdot))$ reduces to
	\begin{equation}
		\begin{aligned}
			\label{eq:F-sym}
			F(w(\cdot))
			&=
			\kappa\!\left(
			\frac{{k_1}'}{\mathrm{Var}(X)}
			+
			\frac{{k_2}'}{\mathrm{Var}(X^2)}
			\right).
		\end{aligned}
	\end{equation}
	Here, $k_1'$ and $k_2'$ denote the coefficients defined in \eqref{eq:positive-coefficients} evaluated at the worst-case source location $\bm{p}_{\mathrm{worst}}$, i.e.,
	\begin{equation}
		\boldsymbol{p}_{\mathrm{worst}} = \arg\max_{\boldsymbol{p}\in\Gamma} \quad \mathrm{SPEB}(w(\cdot), \bm{p}).
	\end{equation}

	In general, the worst-case location(s) $\boldsymbol{p}_{\mathrm{worst}}$ may
	depend on distribution $w(\cdot)$. However, for centro-symmetric distributions,
	$\boldsymbol{p}_{\mathrm{worst}}$ becomes independent of
	$w(\cdot)$ and is located on the broadside of the MA array.
	
	\begin{proposition}
		\label{lem:worstcase-broadside}
		For any centro-symmetric distribution $w(\cdot)$ supported on $[-a,a]$, 
		the worst-case
		source location is independent of $w(\cdot)$, and is
		given by
		\begin{equation}
			\begin{aligned}
				\boldsymbol{p}_{\mathrm{worst}}
				= [\,0,\; d_{\mathrm{R},\max}\,]^{\mathrm{T}},
			\end{aligned}
		\end{equation}
		where $d_{\mathrm{R},\max}=\frac{2D^2}{\lambda}$ denotes the broadside Rayleigh distance.
	\end{proposition}
	
	\begin{IEEEproof}
		See Appendix~\ref{app.proposition2}.
	\end{IEEEproof}

	According to Proposition~\ref{lem:worstcase-broadside}, the coefficients $k_1'$ and $k_2'$ take the specific values of $d_{\mathrm{R},\max}^2$ and $4d_{\mathrm{R},\max}^4$, respectively. Substituting these values back into \eqref{eq:F-sym} yields the simplified worst-case metric.
		\begin{equation}
			\begin{aligned}
				\label{eq:F-sym-broadside}
				F(w(\cdot))
				&=
				\kappa\!\left(
				\frac{d_{\mathrm{R},\max}^2}{\mathrm{Var}(X)}
				+
				\frac{4d_{\mathrm{R},\max}^4}{\mathrm{Var}(X^2)}
				\right).
			\end{aligned}
		\end{equation}

	Accordingly, the optimization problem $(\text{P1}')$ reduces to (P2):
	\begin{subequations}\label{eq:opt-sym-broadside}
		\begin{alignat}{3}
			(\text{P2}):\quad & & \min_{w(\cdot)} & \quad 
			\kappa\!\left( \frac{d_{\mathrm{R},\max}^2}{\mathrm{Var}(X)} + \frac{4d_{\mathrm{R},\max}^4}{\mathrm{Var}(X^2)} \right) \\
			& & \mathrm{s.t.} & \quad \sum_{\zeta\in\mathcal S} w(\zeta)=1, \\ &&& \quad w(\zeta)\ge 0,\quad \forall \zeta\in\mathcal S.
		\end{alignat}
	\end{subequations}
       
	It is worth noting that the objective function of (P2) depends on $w(\cdot)$ solely through the moments $\mathrm{Var}(X)$ and $\mathrm{Var}(X^2)$. This explicit dependence plays a crucial role in the subsequent analysis.
	
	\section{Optimal MA Geometry Analysis via Moment Methods}
	
	In this section, we first employ the moment problem theory to analyze the optimal antenna distribution, proving that the optimum solution is achieved by a specific distribution with at most three points at the center and two edges. Building upon this structural insight, we subsequently develop a discrete closed-form antenna placement strategy that strictly satisfies the minimum inter-element spacing constraint in \eqref{eq. cons-inter-spacing}.
	
	\subsection{Moment-Based Structural Analysis for Symmetric Distribution}

	Denote the second and fourth moments of $X$ as $m_2=\mathbb{E}[X^2]$ and $m_4=\mathbb{E}[X^4]$. Then,
		\begin{equation}
			\begin{aligned}
				\mathrm{Var}(X)=m_2, \qquad
				\mathrm{Var}(X^2)=m_4 - m_2^2,
			\end{aligned}
		\end{equation}

	Hence, problem (P2) can be recast purely in terms of moments $(m_2, m_4)$.
		\begin{equation}
			\begin{alignedat}{3}
				\label{eq:obj-sym-re}
				& & \min_{w(\cdot)} & \quad \kappa\!\left(\frac{d_{\mathrm{R},\max}^2}{m_2}+\frac{4d_{\mathrm{R},\max}^4}{m_4-m_2^2}\right), \\
				& & \mathrm{s.t.}  & \quad \sum_{\zeta\in\mathcal S} w(\zeta)=1, \\ &&&  \quad w(\zeta)\ge 0,\quad \forall \zeta\in\mathcal S.
			\end{alignedat}
		\end{equation}

    Let $\mathcal{M} = \bigl\{m_0,m_2,m_4\bigr\}$ denote the attainable moment set. By inspecting the structure of \eqref{eq:obj-sym-re}, we establish the following
	two structural lemmas for the optimization over $\mathcal M$

	\begin{lemma}
		The minimum of \eqref{eq:obj-sym-re} over $\mathcal M$ can be attained by a centro-symmetric discrete distribution supported on at most three points (locations) of the array.\footnote{Any such centro-symmetric
			three-point distribution must include a location at the origin.}
	\end{lemma}
	
	\begin{IEEEproof}
    See Appendix~\ref{app.lemma2}.
	\end{IEEEproof}
	
	\begin{lemma}
		 The minimum of \eqref{eq:obj-sym-re} over $\mathcal M$ can be achieved by a centro-symmetric discrete distribution supported at the center and two edges of array, i.e., $\{-a,0,+a\}$.
	\end{lemma}

    \begin{IEEEproof}
    	See Appendix~\ref{app.lemma3}.
    \end{IEEEproof}

	\subsection{Closed-Form Optimal Array Design}
	
	By Proposition~1, the search for optimal distribution is restricted to
	centro-symmetric distributions. Furthermore, Lemma~2 and Lemma~3 imply that, within the
	class of centro-symmetric distributions, the optimum can be achieved by a three-point
	distribution supported on $\{-a,0,+a\}$. This yields the
	following theorem.
	
	\begin{theorem}
		\label{thm:three-point-optimal}
		The optimal distribution of \eqref{eq:opt-sym-broadside} is obtained by a three-point distribution supported on $\{-a,0,+a\}$.
	\end{theorem}
	
	Consider such a three-point centro-symmetric distribution supported on
	$\{-a,0,+a\}$ with corresponding probability masses $(\tfrac{q}{2}, 1-q, \tfrac{q}{2})$, where $q \in [0,1]$. The corresponding moments are given by
	\begin{equation}
		\mathrm{Var}(X) = q\,a^2,
		\qquad
		\mathrm{Var}(X^2) = q(1-q)a^4.
	\end{equation}

	Hence, problem \eqref{eq:opt-sym-broadside} becomes a one-dimensional
	minimization over $q$:
		\begin{equation}
			\begin{aligned}
				\label{eq:obj-p}
				\min_{q\in[0,1]}\ \frac{\kappa\,d_{\mathrm{R},\max}^2}{a^2}\left(\frac{1}{q}+\frac{\gamma}{q(1-q)}\right),
			\end{aligned}
		\end{equation}
	where $\gamma = \frac{4d_{\mathrm{R},\max}^2}{a^2} = \frac{256a^2}{\lambda^2}$.
	This is a convex optimization problem and admits a unique minimizer $q^\star\in(\frac{1}{2},1]$ obtained by taking the derivative with respect to $q$, yielding
	\begin{equation}
		\begin{aligned}
			\label{eq:qstar}
			&q^\star \;=\; 1 + \gamma - \sqrt{\gamma(1+\gamma)}.
		\end{aligned}
	\end{equation}

	It is worth noting that $q^\star$ monotonically decreases from $1$ and asymptotically approaches $0.5$ as the value of $a/\lambda$ increases, as shown in Fig.~\ref{fig:pstar}. In practical deployments where the aperture spans several wavelengths ($a \gg \lambda$), $q^\star$ rapidly converges to its asymptotic limit. Hence, we can directly set $q^\star = 0.5$, yielding a fixed probability mass allocation of $(0.25, 0.5, 0.25)$ for an efficient near-optimal array design.
	
		\begin{figure}[t]
		\centering
		\includegraphics[width=0.48\linewidth, keepaspectratio]{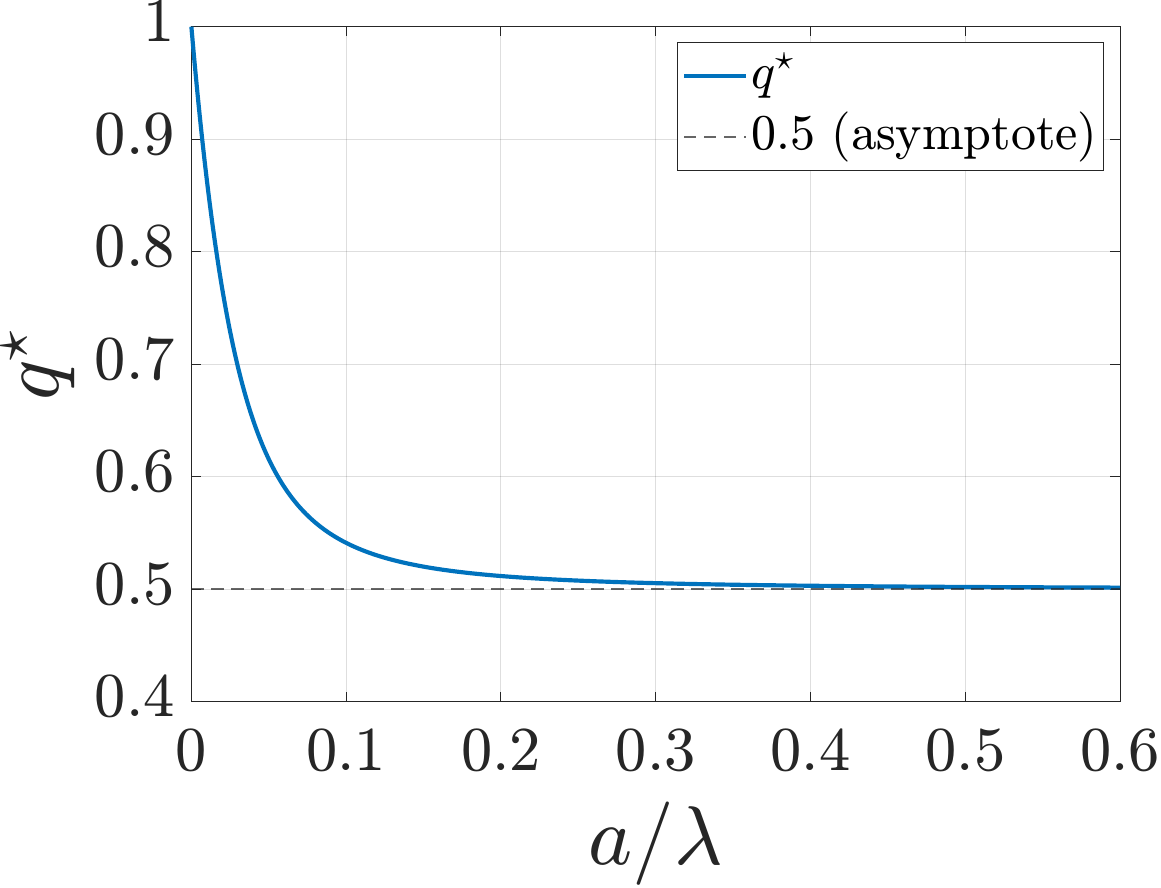}
		     		\caption{Optimal $q^\star$ versus $a/\lambda$.}
		     		\label{fig:pstar}
	\end{figure}
	
	Having derived the closed-form solution for problem $(\text{P1}')$, we now incorporate the inter-element spacing constraint and the given number of antennas $N$ to ensure practical feasibility. Based on the optimal three-point centro-symmetric distribution obtained above, we propose a discrete realization strategy that approximates the theoretical optimum while satisfying the minimum spacing $d$.
	
	Specifically, we allocate $N_{\text{left}} = N_{\text{right}} = \operatorname{round} (0.25N) $ antennas (where $\operatorname{round}(\cdot)$ rounds to the nearest integer) to form clusters at the left and right aperture endpoints, respectively, while the remaining $N_{\text{c}} = N - N_{\text{left}} - N_{\text{right}}$ antennas are placed around the aperture center. Within each of these three clusters (i.e., the center and two edges), the antennas are placed uniformly with the fixed minimum spacing $d$.

	\section{Numerical Results}
	
	This section presents simulation results to validate the effectiveness of the proposed closed-form array design.
	  Specifically, we evaluate its sensing performance against three fixed geometry benchmarks and an exhaustive-search baseline:
	  
	  \begin{itemize} 
	  	\item \textbf{ULA ($\lambda/2$ spacing) \cite{Balanis}:} A conventional uniform linear array with half-wavelength inter-element spacing, centered within the aperture.
	  	\item \textbf{Sparse ULA \cite{ShiTSP2021}:} A uniformly spaced array stretched to span the full aperture, resulting in a larger inter-element spacing than $\lambda/2$.
	  	\item \textbf{Two-edge array \cite{MaTWC0824}:} An array with antennas equally divided and clustered at the two edges of the aperture. 
	  	 \item \textbf{Exhaustive search:} A numerically optimized array obtained via exhaustive searching over all possible antenna locations, serving as a high-complexity benchmark with the best achievable performance.
	  \end{itemize}

	The simulation parameters are set as follows. Carrier frequency $f_0 = 28~\text{GHz}$ \cite{ZhangTWC2022} (i.e., wavelength $\lambda \approx 0.01~\text{m}$), minimum inter-antenna spacing $d = \lambda/2$, and number of snapshots $T = 1024$. Unless otherwise specified, the half-aperture $a = 25\lambda$, number of antennas $N = 25$, and average received signal-to-noise ratio $\text{SNR} = 5~\text{dB}$ \cite{JoJSAC2023}.

	\begin{figure}[t]
		\centering
		
		\includegraphics[width=0.48\linewidth, height=0.48\linewidth, keepaspectratio]{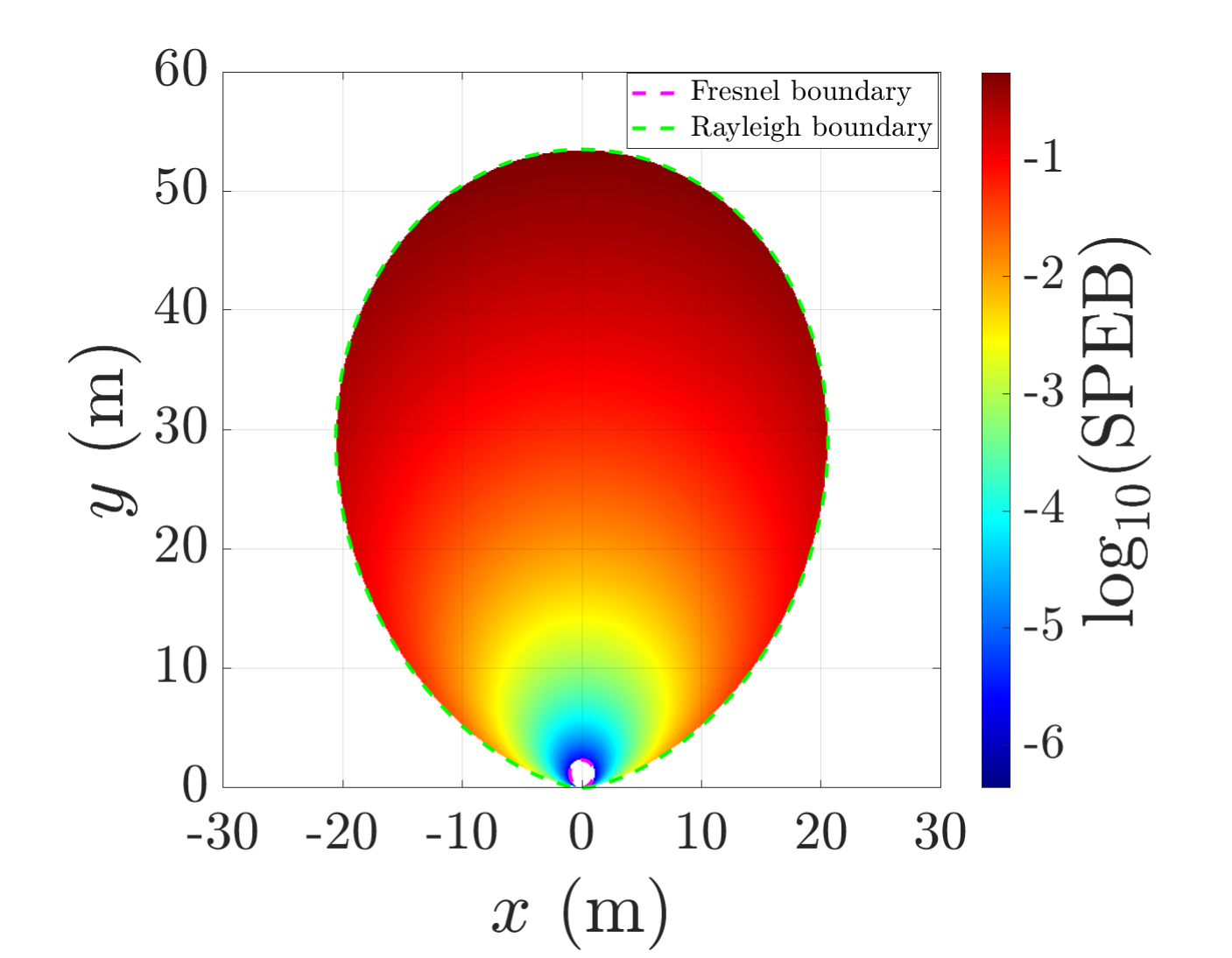}
		
		\caption{The heatmap for near field region.}
		\label{fig:heatmap}
	\end{figure}

	\begin{figure}[t]
		\centering
		\begin{minipage}[b]{0.48\linewidth}
			\centering
			\includegraphics[width=\linewidth]{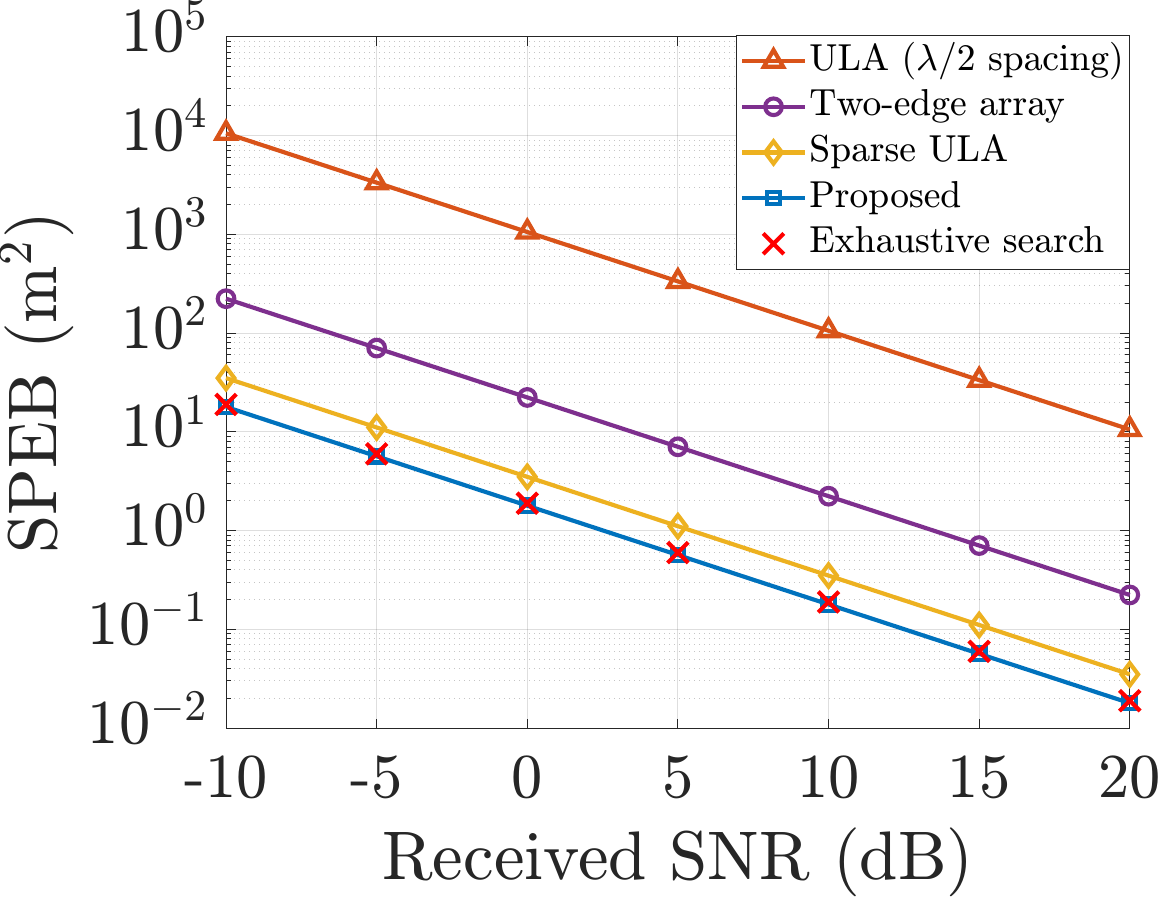}
			\vspace{-5pt}
			\centerline{(a)}
			\label{fig:compare-snr}
		\end{minipage}
		\hfill
		\begin{minipage}[b]{0.48\linewidth}
			\centering
			\includegraphics[width=\linewidth]{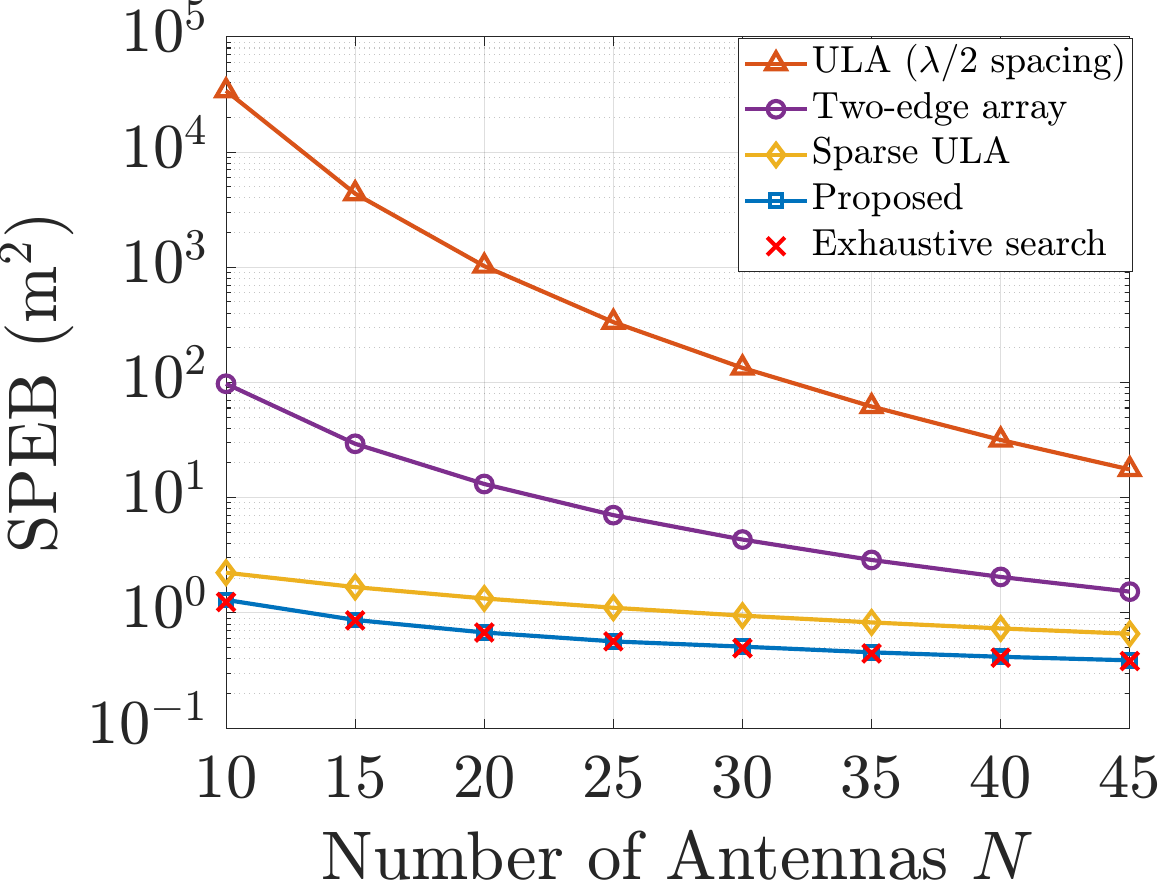}
			\vspace{-5pt}
			\centerline{(b)}
			\label{fig:compare-N}
		\end{minipage}
		\caption{Worst-case $\text{SPEB}$ performance comparison versus (a) Received SNR; (b) Number of antennas $N$.}
		\label{fig:SNR_comparison}
	\end{figure}
	 
	 Fig. \ref{fig:heatmap} illustrates the spatial distribution of $\log_{10}(\mathrm{SPEB})$ across the near-field region $\Gamma$ for the array geometry obtained via the proposed closed-form solution. As shown, the estimation performance is the best (lowest SPEB, indicated by the blue region) near the array origin, while it degrades significantly as the source approaches the outer boundary of $\Gamma$, with the worst performance at the broadside. 
	 This result aligns with Proposition~2, which mathematically demonstrates that the worst-case estimation bound occurs at the broadside.  

	 Fig. \ref{fig:SNR_comparison} evaluates the worst-case $\mathrm{SPEB}$ against variations in received SNR and number of antennas $N$. Across both investigated regimes, the proposed closed-form design consistently achieves the lowest estimation error, significantly outperforming traditional fixed-geometry benchmarks. Notably, our theoretical solution exhibits a remarkable alignment with the exhaustive search baseline (red crosses), corroborating its effectiveness. While increasing SNR and $N$ yields expected performance gains due to enhanced power and additional DoFs. These results prove the proposed design maintains robust superiority, confirming its efficiency in diverse near-field sensing scenarios.
	 
       \section{Conclusion}

	   In this paper, we investigated the optimal design of a linear MA array for near-field single-source sensing. We formulated a worst-case optimization problem with the SPEB as performance metric. By leveraging the Richter-Tchakaloff theorem, we proved that the optimal antenna position distribution admits a centro-symmetric three-point structure. This structural result further enables a closed-form design and an efficient discrete deployment strategy. Simulation results show that the proposed design consistently outperforms conventional uniform distributed antenna configurations and achieves performance comparable to exhaustive search at substantially reduced computational complexity.

		\appendices

		\section{Proof of Lemma~1}
		
		The Cartesian position parameter $\bm{p} = (p_1,p_2)^\mathrm{T}$ is related to
		$\bm{\eta} = (u,r)^\mathrm{T}$ through
		$
				p_1 = r u, 
				p_2 = r\sqrt{1-u^2},
		$
		whose Jacobian transform matrix is
		$\renewcommand{\arraystretch}{0.2} 
		\setlength{\arraycolsep}{2pt}     
					J(\bm{\eta})
					= \frac{\partial \bm{p}}{\partial \bm{\eta}}
					=
					\begin{bmatrix}
						r & u\\
						-\tfrac{ru}{\sqrt{1-u^2}} & \sqrt{1-u^2}
					\end{bmatrix}.
		$
		
		According to \cite{kay1993fssp_estimation}, the CRB matrix w.r.t. parameter $\bm{p}$, is given as
      	$
				\mathrm{CRB}_{\bm{p}}(\bm{x},\bm{\eta})
				= J(\bm{\eta})\,\mathrm{CRB}_{\bm{\eta}}(\bm{x},\bm{\eta})\,J(\bm{\eta})^\mathrm{T}.
		$
		By calculating 
		$
		J(\bm{\eta})^\mathrm{T} J(\bm{\eta})
		=
		\begin{bmatrix}
			\tfrac{r^2}{1-u^2} & 0\\  
			0 & 1
		\end{bmatrix}
		$, 
		we finally obtain
		\begin{small}
			\begin{equation}
				\begin{aligned}
					\mathrm{SPEB}(\bm{x},\bm{p}) 
					&= 
					\tr\!\big(\mathrm{CRB}_{\bm{\eta}}(\bm{x},\bm{\eta})\,J(\bm{\eta})^\mathrm{T} J(\bm{\eta})\big) \\
					&=
					\frac{r^2}{1-u^2}\,\mathrm{CRB}_u(\bm{x})
					+ \mathrm{CRB}_r(\bm{x},\bm{\eta}).
				\end{aligned}
			\end{equation}
		\end{small}
		
		\section{Proof of Proposition~1}

		Define the covariance matrix
		$	\Sigma(w(\cdot)) =
			\begin{bmatrix}
				\mathrm{Var}(X) & \mathrm{Cov}(X,X^2)\\
				\mathrm{Cov}(X,X^2) & \mathrm{Var}(X^2)
			\end{bmatrix}
		$.
			Set
			$
				e_1=[1,0]^\mathrm{T}$ and $
				v(u,r)=[2ur, -2r^2]^\mathrm{T},
			$
			define $L_{u,r}(Z)$ for any positive definite matrix $Z\in\mathbb{S}^2_{++}$:
			\begin{small}
				\begin{equation}
					\begin{aligned}
						L_{u,r}(Z)
						= \kappa\,\frac{r^2}{1-u^2}\,e_1^\mathrm{T} Z e_1
						+\frac{\kappa}{(1-u^2)^2}\,v(u,r)^\mathrm{T} Z v(u,r)
					\end{aligned}
				\end{equation}
			\end{small}

			Then, $F(w(\cdot))=\max_{(u,r)\in\Gamma} L_{u,r}\big(\Sigma(w(\cdot))^{-1}\big)$, 
			let $\check w(\zeta)=w(-\zeta)$ and $w_{\mathrm{sym}}(\zeta)=\tfrac12\big(w(\zeta)+w(-\zeta)\big)$.
			A straightforward computation gives
					$\Sigma(w_{\mathrm{sym}}(\cdot))
					\succeq\tfrac12\big(\Sigma(w(\cdot))+\Sigma(\check w(\cdot))\big)$,

			Since inverse operation is monotone decreasing and convex on
			$\mathbb{S}^2_{++}$, so
			$
				\Sigma(w_{\mathrm{sym}}(\cdot))^{-1}
				\preceq
				\tfrac12\big(\Sigma(w(\cdot))^{-1}+\Sigma(\check w(\cdot))^{-1}\big).
			$
			We note that map $Z$ to $L_{u,r}(Z)$ is linear and Loewner-monotone on $\mathbb{S}^2_{++}$.
			Hence we accordingly obtain
			$
				L_{u,r}\big(\Sigma(w_{\mathrm{sym}}(\cdot))^{-1}\big)
				\le
				\frac12\Big(
				L_{u,r}\big(\Sigma(w(\cdot))^{-1}\big)
				+L_{u,r}\big(\Sigma(\check w(\cdot))^{-1}\big)
				\Big),
			$
			which implies
				$
					F(w_{\mathrm{sym}}(\cdot))
					\le
					\frac12\big(F(w(\cdot))+F(\check w(\cdot))\big).
				$
				
			Since $\Gamma$ is symmetric,
			the change of variable $(u,r)$ to $(-u,r)$ leaves the maximization set
			invariant, and therefore
			$
				F(\check w(\cdot)) = F(w(\cdot)).
			$
			Combining this with the inequality above gives
			$
				\frac12\big(F(w(\cdot))+F(\check w(\cdot))\big) = F(w(\cdot)).
			$
			Thus $F(w_{\mathrm{sym}}(\cdot))\le F(w(\cdot))$, which proves Proposition~1.
			\label{app.proposition1}

		\section{Proof of Proposition~2}

			It can be observed \eqref{eq:opt-sym-broadside}
			is strictly increasing in $r$ for any $u\in[-1,1]$. This indicates the maximum over the near-field
			region is attained on the Rayleigh boundary
			$r=d_R(u)$. We define
			$
					T(u;w(\cdot))
					=
					\kappa\!\left[
					\frac{k_1}{\mathrm{Var}(X)}
					+
					\frac{k_2}{\mathrm{Var}(X^2)}
					\right].
			$
			Substituting $r=d_{\mathrm{R}}(u) = d_{\mathrm{R}, \max}(1-u^2)$ and the coefficients $k_1, k_2$ from (\ref{eq:positive-coefficients}) into the $T(u;w(\cdot))$ yields
			\begin{small}
				\begin{equation}
					\begin{aligned}
						\label{App. eq. T(u)}
						\frac{1}{\kappa}T(u;w(\cdot))
						=
						\frac{d_{\mathrm{R}, \max}^2(1+3u^2)}{\mathrm{Var}(X)}
						+
						\frac{4d_{\mathrm{R}, \max}^4(1-u^2)^2}{\mathrm{Var}(X^2)}.
					\end{aligned}
				\end{equation}
			\end{small}
			
			Set $\upsilon=u^2\in[0,1]$ and rewrite \eqref{App. eq. T(u)} as $f(\upsilon)$. Then
			$   f''(\upsilon)  =   \frac{8d_{\mathrm{R}, \max}^4}{\mathrm{Var}(X^2)}>0   $.  So $f$ is strictly convex on $[0,1]$, and hence its maximum on $[0,1]$ is attained at an endpoint $\upsilon\in\{0,1\}$, i.e., $u_\mathrm{worst}^2\in\{0,1\}$. Then, compare the value of two endpoints
			\begin{small}
				\begin{equation}
					\begin{aligned}
						\frac{1}{\kappa}\left( T(0;w(\cdot))-T(1;w(\cdot)) \right)
						= 4 d_{\mathrm{R}, \max}^2 \left( \frac{d_{\mathrm{R}, \max}^2}{\mathrm{Var}(X^2)} - \frac{3}{4\mathrm{Var}(X)} \right).
					\end{aligned}
				\end{equation}
			\end{small}

			For practical MA arrays systems, the aperture typically spans multiple wavelengths, yielding $d_{\mathrm{R}, \max} > \frac{\sqrt{3}}{2}a$.
			Therefore
				\begin{equation}
					\begin{aligned}
						\operatorname{Var}(X^2)\le
						\mathbb{E}[X^4]
						\le
						a^2 \mathbb{E}[X^2]
						=
						a^2\mathrm{Var}(X) < \frac{4}{3} d_{\mathrm{R}, \max}^2\mathrm{Var}(X).
					\end{aligned}
				\end{equation}
			This indicates $T(0;w(\cdot))> T(1;w(\cdot))$, so the maximum is
			attained at $(u, r)=(0, d_{\mathrm{R},\max})$.
			Hence the worst-case location
			$
			\boldsymbol{p}_{\mathrm{worst}}
			= [\,0,\; d_{\mathrm{R},\max}\,]^{\mathrm{T}}
			$, which proves Proposition~2.
			
			\label{app.proposition2}
        
        \section{Proof of Lemma~2}
        		
        We first invoke Richter-Tchakaloff Theorem \cite{schmüdgen2017}. It says: Suppose that $(\mathcal{Y},\mu)$ is a measure space, 
        and let 
        $V \subset L^1_{\mathbb{R}}(\mathcal{Y},\mu)$ be a finite-dimensional 
        linear subspace of finite dimension $\dim(V)$, where $L^1_{\mathbb{R}}(\mathcal{Y},\mu)$ denotes the space of real-valued integrable functions on $\mathcal{Y}$.
        Let $L^\mu:V\to\mathbb{R}$ be the moment functional defined by
        $
        			L^\mu(f) = \int f 
       $ for $f\in V$.
        Then there exists a $k$-atomic measure $\nu = \sum_{j=1}^k z_j \delta_{x_j}$ with $k \le \dim(V)$, 
        such that $L^\mu = L^\nu$.
        
        Now we apply this theorem to prove Lemma~2.
        Observing \eqref{eq:obj-sym-re} depends on the measure $\mu$ solely through the moments $(m_0, m_2, m_4)$. Let $V = \operatorname{span}\{1, x^2, x^4\} \subset L^1_{\mathbb{R}}([-a,a])$ be the subspace with $\dim(V) = 3$. According to the Richter-Tchakaloff Theorem, for any arbitrary probability measure $\mu$ on $[-a,a]$, there exists a $k$-atomic measure $\nu$ with $k \le \dim(V) = 3$ such that $\nu$ preserves the moments of $\mu$, the objective values for both measures are identical:
        	\begin{equation}
        		\begin{aligned}
        			\mathcal{J}(\nu)
        			&= \varphi \big(m_0(\nu),m_2(\nu),m_4(\nu)\big) \\
        			&= \varphi \big(m_0(\mu),m_2(\mu),m_4(\mu)\big)
        			= \mathcal{J}(\mu).
        		\end{aligned}
        	\end{equation}
        where $\mathcal{J}(\cdot)$ denotes the objective functional defined on the measure space, and $\varphi(\cdot)$ represents the mapping from the moment vector $(m_0, m_2, m_4)$ to the objective value.
        This indicates that the performance achievable by a general probability measure can be equally attained by a discrete measure supported on at most three points. Therefore we can restrict the optimization domain to the set of measures with at most three atoms. This completes the proof.
        \label{app.lemma2}
        
        \section{Proof of Lemma~3}
			Let $\mu$ be a probability measure supported on $[-a,a]$. Since $x^4 \le a^2 x^2$ for all $x\in[-a,a]$, we have
			$
					m_4 
					\le 
					 a^2 m_2.
			$
			By the Cauchy-Schwarz inequality
				\begin{equation}
					m_2^2
					= \Bigl(\int x^2 d\mu(x)\Bigr)^2
					\le \Bigl(\int x^4 d\mu(x)\Bigr)\Bigl(\int 1\,d\mu(x)\Bigr)
					= m_4.
				\end{equation}
			Thus any feasible pair $(m_2,m_4)$ satisfies
			$
			m_2^2 \le m_4 \le a^2 m_2.
			$
			
			For each fixed $m_2$, the objective in \eqref{eq:obj-sym-re} is strictly
			decreasing in $m_4$ on $[m_2^2,a^2 m_2]$, so
			for each $m_2$ there exists an optimal solution lying on the boundary $m_4 = a^2 m_2$.
			Moreover, equality in $m_4 
			\le 
			a^2 m_2$ holds if and only if
			$x^4 = a^2 x^2$ for all $x$ in the support of $\mu$, which implies
			$x^2 \in \{0,a^2\}$ and hence $x \in \{-a,0,+a\}$. Together with Lemma~2,
			the global optimum can be attained by a centro-symmetric three-point
			measure supported on $\{-a,0,+a\}$, proving Lemma~3.
             \label{app.lemma3}
		
		\bibliographystyle{IEEEtran}
		\bibliography{references}   
		
	\end{document}